\documentclass[12pt,draftcls,onecolumn]{IEEEtran}

\normalsize

\ifCLASSINFOpdf
\else
\fi

\usepackage{stfloats}

\usepackage{mathtools}
\usepackage{amsmath}
\usepackage{amssymb}
\usepackage{amsmath}
\usepackage{cite}
\usepackage{siunitx}
\usepackage{mathtools}
\usepackage{amsmath}
\usepackage{amssymb}
\usepackage{amsmath}
\usepackage{cite}
\usepackage{siunitx}

\usepackage{amsthm}
\usepackage{graphicx}
\usepackage{graphicx}
\usepackage[utf8]{inputenc}

\usepackage{amsthm}

\usepackage{epstopdf}
\usepackage{mdwmath}
\usepackage{blindtext}
\usepackage{eqparbox}
\usepackage{fixltx2e}

\usepackage{graphicx}
\usepackage{graphicx}

\usepackage{mdwmath}
\usepackage{eqparbox}
\usepackage{fixltx2e}
\DeclareMathOperator*{\argmin}{\arg\!\min}
\DeclareMathOperator*{\argmax}{\arg\!\max}

\begin{document}
\title{Optimal Transmission Policies for Multi-hop Energy Harvesting Systems} %
\vspace{0ex}
\author{\normalsize Milad Rezaee, Mahtab~Mirmohseni, Vaneet Aggarwal, \textit{Senior Member, IEEE}, and~Mohammad Reza Aref\\
\thanks{M. Rezaee, M. Mirmohseni and M. R. Aref are with the Information Systems and Security Laboratory,
Department of Electrical Engineering, Sharif University of Technology,
Tehran 11365/8639, Iran (e-mail: miladrezaee@ee.sharif.edu; mirmohseni@
sharif.edu; aref@sharif.edu). }
\thanks{V. Aggarwal is with Purdue University, West Lafayette, IN 47907 USA
(e-mail: vaneet@purdue.edu).}
}

\markboth{}%
{}
\maketitle

\begin{abstract}
\boldmath
In this paper, we consider a \emph{multi-hop} energy harvesting (EH) communication system in a full-duplex mode, where arrival data and harvested energy curves in the source and the relays are modeled as \emph{general} functions. This model includes the EH system with discrete arrival processes as a special case. We investigate the throughput maximization problem considering minimum utilized energy in the source and relays and find the optimal offline algorithm. We show that the optimal solution of the two-hop transmission problem have three main steps: (i) Solving a point-to-point throughput maximization problem at the source; (ii) Solving a point-to-point throughput maximization problem at the relay (after applying the solution of first step as the input of this second problem); (iii) Minimizing utilized energy in the source. In addition, we show that how the optimal algorithm for the completion time minimization problem can be derived from the proposed algorithm for throughput maximization problem. Also, for the throughput maximization problem, we propose an online algorithm and show that it is more efficient than the benchmark one (which is a direct application of an existing point-to-point online algorithm to the multi-hop system).
\end{abstract}


\IEEEpeerreviewmaketitle
\section{Introduction}
%\fontsize{10}{11} \selectfont
\newtheorem{conj}{Conjecture}
Nowadays the wireless nodes with a limited battery have limited lifetimes which can be extended using energy harvesting (EH) from renewable sources such as vibration absorption devices, water mills, wind turbines, microbial fuel cells and solar cells and using rechargeable batteries. The recent progress in technology has made possible the design of EH devices with sufficient power which is required for communication objectives. Although, EH sources offer an unbounded energy supply to be harvested, the stochastic nature of EH makes using this energy for communication difficult. It is thus essential that this energy is used in an efficient way.

Optimal scheduling in EH systems has two main categories: throughput maximization problem and completion time minimization problem. Focusing on the first problem, one of the important research fields in this area looks for optimal schemes to maximize the throughput in a given deadline in a point-to-point channel \cite{ozel2011transmission,bodin2014energy,wang2014power,wang2013renewable,wang2015iterative}. This problem for an AWGN fading channel is studied in \cite{ozel2011transmission}. \cite{bodin2014energy} considers this problem when the battery is limited and the optimal power can be chosen from a finite set of real numbers.  \cite{wang2014power} formulates a throughput maximization problem as a Markov decision process and proposes an algorithm with lower computational complexity than the standard discrete Markov decision process method. The authors of \cite{wang2013renewable,wang2015iterative} consider this problem with finite battery capacity and maximum transmission power, and propose an energy efficient dynamic-waterfilling algorithm. Also, \cite{arafa2014single} investigates a throughput maximization problem with EH transmitter (Tx) and receiver (Rx) while the Rx utilizes the harvested energy for the decoding process.
The second problem, which aims to minimize the completion time to transmit a given amount of data, is studied in \cite{yang2012optimal,ozcelik2012minimization}. A point-to-point channel with an EH Tx and arrival data during the transmission process has been considered in \cite{yang2012optimal}. The authors in \cite{ozcelik2012minimization} have continued the work of \cite{yang2012optimal} by considering a fading channel.

In traditional centralized communications schemes, each user connects to the nearest base station. Nevertheless, multi-hop relay connectivity structure is foreseen to be a revolutionary way of connections in the design of 5G.
This, besides the centralized communications schemes, makes the new model of Device-to-Device networking, in which each user is permitted to communicate peer-to-peer by use of direct links \cite{kim2015survey}. Also, in millimeter-wave communications, which is a promising technology for high rate multimedia applications, millimeter-wave signal power diminishes extremely over distance due to propagation loss at high-frequency. Hence, using multiple short hops can be more efficient than one long hop \cite{qiao2011enabling}.
These scenarios motivate the design of multi-hop (relay) networks, where they can be used with highest efficiency in the full-duplex (FD) mode with relays transmitting and receiving at the same time/frequency band. In fact, from a physical point of view, considering a point-to-point communication with perfect self-interference cancellation in a full-duplex mode doubles the spectral efficiency of half-duplex (HD) mode.
However, due to self-interference caused by the relay's transmitter on its receiver, relays work in the HD mode in traditional multi-hop relay networks. Recently, using advanced antennas and digital baseband technologies along with radio frequency (RF) interference cancellation techniques, self-interference is decreased near to the level of the noise floor in low-power networks.  As a result of this promising attribute, FD technology is rapidly developing its applications in wireless communications \cite{jain2011practical,duarte2014design}.

Considering a \emph{discrete} harvested energy arrival process, the two-hop EH systems (with no data arrival process) is studied in \cite{gunduz2011two,luo2013optimal,tutuncuoglu2013optimum,orhan2015energy}. The authors in \cite{gunduz2011two} propose the optimal offline algorithm in a throughput maximization problem in FD mode and HD mode with EH only at the relay. An HD two-hop relay channel with EH only at the source has been considered in \cite{luo2013optimal} for both throughput maximization and completion time minimization problems. The authors of \cite{tutuncuoglu2013optimum} have studied a two-way relay channel with EH nodes. In addition, a throughput maximization problem with HD relay nodes which have limited data buffers has been considered in \cite{orhan2015energy}. In \cite{gurakan2014energy}, the authors consider a throughput maximization problem for a diamond channel with one-way energy transfer from Tx node to the relays.
Except a part in \cite{gunduz2011two}, the relay in all of the above papers, works in an HD mode. Considering the ability of simultaneous transmission/reception in an FD relay changes the nature of the data arrival process in the relay node. In fact, in an FD relay, the data arrival process (sent from the source) is no longer discrete, while in an HD relay the data is available at the beginning of the transmission and thus no data arrival process exists. This continuous data arrival process (which is a must for an FD relay) has not been considered in \cite{gunduz2011two}. Thus, in an FD multi-hop system the arrival data process at the relays are continuous.

To the best of our knowledge, the continuous model for the energy arrival process has been investigated in \cite {devillers2012general,varan2014energy,rezaee2015optimal} for the point-to-point channel, where the continuous model for the data arrival process only is studied in \cite {rezaee2015optimal}. In \cite{devillers2012general}, a throughput maximization problem with battery imperfections has been studied. The authors in \cite{varan2014energy} investigate an EH system with a degrading battery of finite capacity by convex analysis tools for a continuous harvested energy curve. Although, in \cite{devillers2012general,varan2014energy}, the harvested energy curve is continuous but all data is stored in information buffer at the beginning of the transmission and the arrival data curve has not been considered. In \cite{rezaee2015optimal}, adding the arrival data (not buffered) to the model, a throughput maximization problem while minimizing the utilized energy in the source has been investigated.
As described, the main motivation for the continuous data arrival comes from the relaying structure. In general, in a throughput maximization problem at a multi-hop relay channel, the arrival data curve at the relay node may be continuous even when the data arrival is discrete at the source.
In addition, considering rateless codes cancels the necessity of packetizing data in some applications \cite{palanki2004rateless} and consequentially the continuous model better fits these cases. Moreover, one another motivation comes from calculus network \cite{le2001network}. To sum up, considering a system with continuous data arrival in a multi-hop setup is significant when analyzing EH systems. Besides, since the amount of harvested energy in an energy harvester device is naturally continuous by time, considering a continuous-time model will be more appropriate for the amount of harvested energy \cite{varan2014energy,ottman2002adaptive}. Although considering a discrete model results in a more tractable problem, the derived optimal scheme for such a model is a suboptimal scheme and it decreases the efficiency as it is shown numerically in Section V.

In this paper, we consider a \emph{multi-hop} EH communication system in an FD mode, in which arrival data and harvested energy curves in the source and the relays are modeled as \emph{general} functions, which subsumes both discrete and continuous models.
Also, we assume that the size of the energy and data buffers at the source, relays, and receiver are infinite. We investigate the throughput maximization problem and find the optimal offline algorithm.
We show that the optimal policy for the throughput maximization problem in a two-hop channel is given by solving two point-to-point throughput maximization problems and then minimizing utilized energy in the source: (i) the first is to maximize the data from the source to the relay; (ii) the second is to maximize the data from the relay to the receiver while considering the optimal transmitted data from the first hop as the arrival data in the relay; (iii) then, minimizing the utilized energy in the source such that the received data at the receiver is kept fixed (as step (ii)).
In addition, we show that how the optimal algorithm for the completion time minimization problem can be derived from the proposed algorithm for throughput maximization problem. Also, we propose an online algorithm to maximize the throughput to the Rx and show that it is more efficient than the benchmark one (which is a direct application of an existing point-to-point online algorithm to the multi-hop system). Finally, the obtained results in this paper are investigated numerically.

In addition to the online algorithm and time minimization problem considered in our paper, two main differences with the only work on the FD two-hop system in \cite{gunduz2011two} on the throughput maximization problem are continuous arrival processes and energy minimization.
In a two-hop channel where the harvested energy curve in the relay is the bottleneck (compared to the harvested energy in the source), transmitting with high energy in the source results in a long data queue in the relay (because of its energy shortage) and does not improve the throughput and/or delay. Thus, if we do not consider the minimization of energy in the source, we may waste a lot of energy. The same discussion can be made for a multi-hop channel with minimizing the energy at the source and the relays.
Also, it is worth noting that by considering the continuous energy and data arrivals, the problem enters a new space where the existing discrete-space proofs are no longer applicable.
We remark that, to the best of our knowledge, even the model with discrete data arrival (not buffered) and discrete energy has not been considered in the previous works. The most challenging parts in this paper are to apply data and energy causalities in continuous space and considering energy minimization in the source, which needs totally different approaches from the discrete model in \cite{gunduz2011two}.

 %In addition, we provide two sets of examples to assess our results numerically. Also, we show that discretizing harvested energy reduce the efficiency of system.

 \theoremstyle{definition}
 \newtheorem{defn}{Definition}[section]
\section{System Model}
We consider a multi-hop wireless communication system, where the source and the relays harvest energy from external sources in a \emph{general} fashion. Also, we assume that at time $t$, the available data at the source is a general function of $t$. To make the problem easier to understand, we first consider a two-hop channel. Finally, in corollaries \ref{1} and \ref{2}, we extend the results to the multi-hop channel. The receiver is assumed to have enough energy to provide adequate power for decoding at any rate that can be achieved by the source and relays. We have the following assumptions for the two-hop channel.

  \begin{figure}
    \centering
     \includegraphics[trim=0cm .5cm 0cm .5cm, width=3.3in,angle=0,clip]{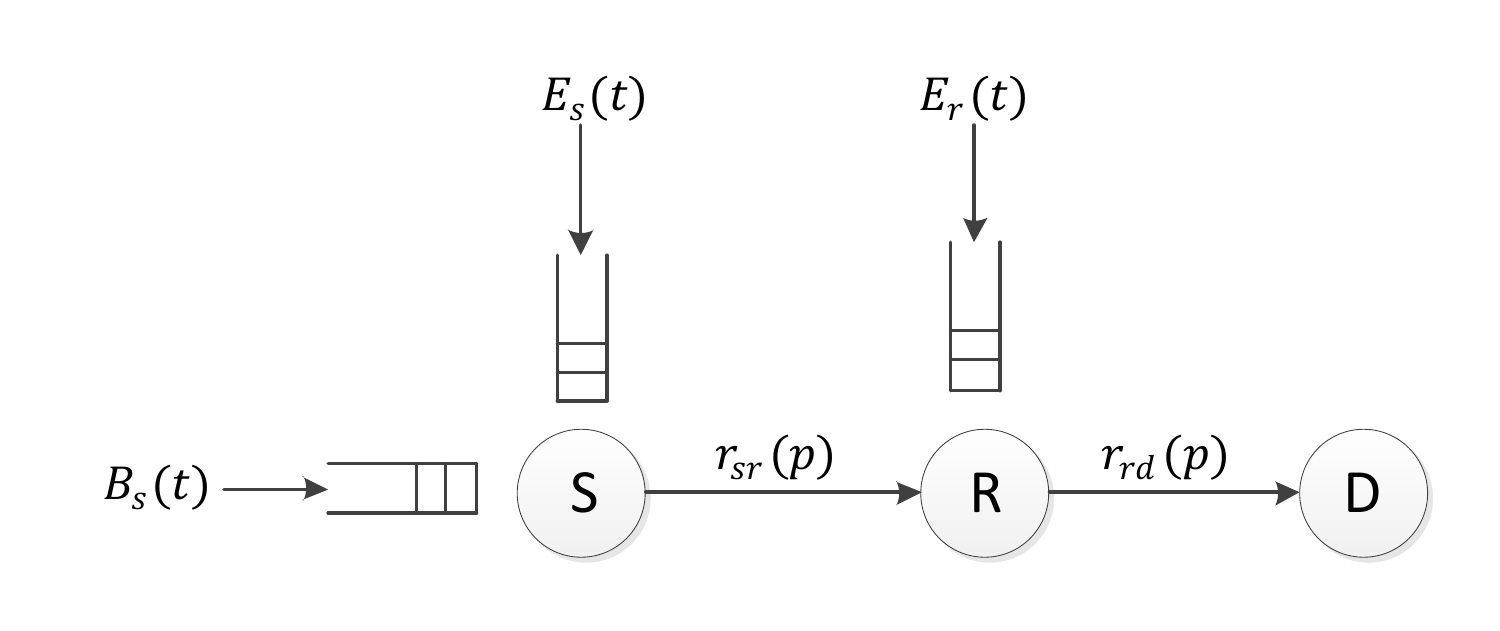}
    \caption{The topology of the network for a two-hop channel}
    \end{figure}

As shown in Fig. 1, $E_s(t)$ and $E_r(t)$ denote the amount of harvested energy at the source and relay in $[0,t]$, respectively. $B_{s}(t)$ denotes the amount of arrived data at the source in $[0,t]$. In this model, to take the general arrival processes into account, $E_{s}(t)$, $E_{r}(t)$ and $B_{s}(t)$ are\emph{ piecewise} continuous functions (which include both discrete and continuous models). Also, we assume that $E_{s}(t)$, $E_{r}(t)$ and $B_{s}(t)$ are bounded. These curves are differentiable functions of $t$, for $t \in [0,\infty)$, except probably in a finite number of points (in these points, $E_{s}(t)$, $E_{r}(t)$ and $B_{s}(t)$ can have discontinuity or unequal right and left derivatives). Moreover, the derivative of $E_{s}(t)$, $E_{r}(t)$ and $B_{s}(t)$ are assumed to be piecewise continuous and bounded (except probably in a finite number of points). $B_{sr}(t)$ and $B_{rd}(t)$ (transmitted data curves in the source and relay, respectively) are the amounts of data which are transmitted from the source to the relay and from the relay to the receiver, respectively. $E_{sr}(t)$ and $E_{rd}(t)$ (transmitted energy curves in the source and relay, respectively) are the amounts of energy that are utilized in the source and the relay to transmit data from the source to the relay, and the relay to the receiver in $[0,t]$ respectively.  $E_{sr}(t)$, $E_{rd}(t)$, $B_{sr}(t)$, and $B_{rd}(t)$ are continuous for $t\in [0,\infty)$ and they are differentiable functions for $t\in [0,\infty)$ (except probably in a finite number of points). $p_{sr}(t)$ and $p_{rd}(t)$ (transmitted power curves in the source and relay, respectively) are the amounts of power used in the source and the relay for data transmission, respectively, which are piecewise continuous. The instantaneous transmission rates in both source and relay relate to the power of transmission through continuous functions $r_{sr}(p)$ and $r_{rd}(p)$, respectively, where the rate functions are concave, monotonically increasing, $\lim\limits_{p\to \infty}r_{sr}(p)=\lim\limits_{p\to \infty}r_{rd}(p)=\infty$ and $r_{sr}(0)=r_{rd}(0)=0$.

 Now, we formulate our problem as follows:
 \begin{eqnarray}
 D^{(MH)}(T)&=&\!\!\!\!\!\max_{p_{sr}(t),p_{rd}(t)} \int_{0}^{T}r_{rd}(p_{rd}(t))dt\label{42}\\
  s.t.~~\int_{0}^{t}p_{sr}(t^{'})&\leq &\!\!\!\!\! E_{s}(t),~0\leq t\leq T\label{43}\\
 \int_{0}^{t}p_{rd}(t^{'})&\leq &\!\!\!\!\! E_{r}(t),~0\leq t\leq T\label{44}\\
  \int_{0}^{t}r_{sr}(p_{sr}(t^{'}))dt^{'}&\leq &\!\!\!\!\! B_{s}(t),~0\leq t\leq T\label{45}\\
    \int_{0}^{t}r_{rd}(p_{rd}(t^{'}))dt^{'}&\leq &\!\!\!\!\!\!\!\int_{0}^{t}r_{sr}(p_{sr}(t^{'}))dt^{'},~0\leq t\leq T\label{46}.
 \end{eqnarray}
\eqref{43} and \eqref{44} are the energy causality conditions in the source and the relay. \eqref{45} and \eqref{46} are the data causality conditions in the source and relay.
Also, we include another condition to make the solution unique: in the cases where the optimal policy is not unique, we choose the policy, among the throughput maximizing policies, that minimizes the utilized energy in the source and relay.

We denote $B_{sr,s}^{*}(t)$ and $E_{sr,s}^{*}(t)$ as the optimal transmitted data and energy curves, respectively, which are obtained from considering the point-to-point throughput maximization problem in the source when problem is to maximize the amount of data from the source to the relay (considering the first hop, only we have the causality conditions in the source). In other words,
\begin{align}\nonumber
E_{sr,s}^{*}(t)=\int_{0}^{t}p_{sr,s}^{*}(t^{'})dt^{'},~B_{sr,s}^{*}(t)=\int_{0}^{t}r_{sr}(p_{sr,s}^{*}(t^{'}))dt^{'},
\end{align}
where $p_{sr,s}^{*}(t)=\argmax\limits_{p_{sr,s}(t)}\{\int_{0}^{T}r_{sr}(p_{sr,s}(t))dt\}$, subject to constraints \eqref{43} and \eqref{45}. $B_{rd,B_{sr}}^{*}(t)$, $E_{rd,B_{sr}}^{*}(t)$ and $p_{rd,B_{sr}}^{*}(t)$ are respectively the optimal transmitted data curve, optimal transmitted energy curve and optimal transmitted power curve which are obtained from the point-to-point throughput maximization problem in the relay to the receiver when we have $B_{sr}(t)$ and $E_{r}(t)$ as arrival data and harvested energy curves in the relay, respectively. This point-to-point problem only considers second hop while the received data from the source is modeled as an arrival data process at the relay. In other words,
\begin{align}\nonumber
E_{rd,B_{sr}}^{*}(t)=\int_{0}^{t}p_{rd,B_{sr}}^{*}(t^{'})dt^{'},\\\nonumber B_{rd,B_{sr}}^{*}(t)=\int_{0}^{t}r_{rd}(p_{rd,B_{sr}}^{*}(t^{'}))dt^{'},
\end{align}
where $p_{rd,B_{sr}}^{*}(t)=\argmax\limits_{p_{rd}(t)}\{\int_{0}^{T}r_{rd}(p_{rd}(t))dt\}$ subject to constraints \eqref{44} and \eqref{46}.
Moreover, we denote $B_{s}^{*}(t)$ and $E_{s}^{*}(t)$ as the optimal transmitted data and energy curves in the source, $B_{r}^{*}(t)$ and $E_{r}^{*}(t)$ are the optimal transmitted data and energy curves in the relay for problem \eqref{42}-\eqref{46}. In other words,
\begin{align}\nonumber
E_{s}^{*}(t)=\int_{0}^{t}p_{s}^{*}(t^{'})dt^{'},~B_{s}^{*}(t)=\int_{0}^{t}r_{sr}(p_{s}^{*}(t^{'}))dt^{'}\\\nonumber
E_{r}^{*}(t)=\int_{0}^{t}p_{r}^{*}(t^{'})dt^{'},~B_{r}^{*}(t)=\int_{0}^{t}r_{rd}(p_{r}^{*}(t^{'}))dt^{'},\nonumber
\end{align}
where $[p_{s}^{*}(t), p_{r}^{*}(t)] =\argmax\limits_{p_{sr}(t),p_{rd}(t)} \{\int_{0}^{T}r_{rd}(p_{rd}(t))dt\}$ subject to \eqref{43}-\eqref{46}.
% Note that $B_{s}^{*}(t)$, $E_{s}^{*}(t)$, $B_{r}^{*}(t)$ and $E_{r}^{*}(t)$ are the optimal curves for the general problem in \eqref{42}-\eqref{46}, in comparison with the optimal curves of one-hop problems which are $B_{sr,s}^{*}(t)$ and $E_{sr,s}^{*}(t)$ for the first hop and the $B_{rd,B_{sr}}^{*}(t)$ and $E_{rd,B_{sr}}^{*}(t)$ for the second hop.

 \newtheorem{theorem}{Theorem}[section]
   \newtheorem{corollary}{Corollary}[theorem]
   \newtheorem{lemma}[theorem]{Lemma}
\newtheorem{rem}{Remark}
\newtheorem{coro}{Corollary}

 \section{Offline Algorithms}
 In this section, we consider a multi-hop channel with one source, one receiver, and many relays and we investigate the throughput maximization and completion time minimization problems in an offline model in an FD mode. For simplicity, we first assume that we have a two-hop communication channel which is illustrated in Fig. 1. Then we extend the results to $n$ relays in Corollaries 1 and 2.
 \subsection{Throughput Maximization}
In the following theorem, we show that the optimal solution of the two-hop transmission problem in \eqref{42}-\eqref{46} is derived by first solving a point-to-point throughput maximization problem at the source, and next solving a point-to-point throughput maximization problem at the relay (after applying the first solution as the input of the second problem).
\begin{theorem}\label{V.1.}
(Relay optimal policy) In the optimal policy, we have $B_{r}^{*}(t)=B_{rd,B_{sr,s}^{*}}^{*}(t)$.
 \end{theorem}
 \begin{proof}[Outline of the proof]
We will use proof by contradiction to show that $B_{rd,B_{sr}}^{*}(t)\leq B_{rd,B_{sr,s}^{*}}^{*}(t)$, which is stronger than the claim of Theorem \ref{V.1.}. The reason is as follows. The received data at the relay (from the source) can be modeled with an arrival data curve in the relay. Thus, knowing this curve (e.g. $\tilde{B}_{sr}(t)$), the optimal strategy at the relay is a solution of the second hop's problem \cite{rezaee2015optimal} (i.e., $B_{rd,\tilde{B}_{sr}}^{*}(t)$). Therefore, the question reduces to find the optimal $\tilde{B}_{sr}(t)$. Computing $B_{rd,B_{sr}}^{*}(t)\leq B_{rd,B_{sr,s}^{*}}^{*}(t)$ at $t=T$ says that the maximum data can be transmitted to the receiver is $B_{rd,B_{sr,s}^{*}}^{*}(T)$ and thus $B_{r}^{*}(t)=B_{rd,B_{sr,s}^{*}}^{*}(t)$.

For our proof by contradiction, we assume $a$ is the first point, for which there is $b>a$ that $\forall t\in(a,b)$, $B_{rd,B_{sr}}^{*}(t)> B_{rd,B_{sr,s}^{*}}^{*}(t)$ holds. Since the transmitted power curve is piecewise continuous, a subinterval of $(a,b)$ exists such that $p_{rd,B_{sr,s}^{*}}^{*}(t)<p_{rd,B_{sr}}^{*}(t)$. Next, by using the technique of \cite [Remark 21]{rezaee2015optimal}, we show that for any arbitrary point such as $t_0$ in the mentioned subinterval, we have $p_{rd,B_{sr}}^{*}(t_0)\leq p_{rd,B_{sr,s}^{*}}^{*}(t_0)$, which is a contradiction.
Now, we show the energy optimality of $B_{rd,B_{sr,s}^{*}}^{*}(t)$. If there is a feasible transmitted data curve $B_{rd,\bar{B}_{sr}}^{*}(t)$ such that $\int_{0}^{T}(B_{rd,B_{sr,s}^{*}}^{*}(t)-B_{rd,\bar{B}_{sr}}^{*}(t))^{2}dt\neq 0$ and $B_{rd,B_{sr,s}^{*}}^{*}(T)=B_{rd,\bar{B}_{sr}}^{*}(T)$, and since $B_{rd,\bar{B}_{sr}}^{*}(t)\leq B_{rd,B_{sr,s}^{*}}^{*}(t)$ for $t\in [0,T]$ holds, then, based on Lemma \ref{zafar}, we have $E_{rd,B_{sr,s}^{*}}^{*}(T) < E_{rd,\bar{B}_{sr}}^{*}(T)$ that contradicts our energy efficient policy.
For the detailed proof, see Appendix \ref{throughputtwohop}. \end{proof}

Next remark shows that the above result is not trivial.
\begin{rem}
In Theorem \ref{V.1.}, we prove that finding the optimal policy for a throughput maximization problem from the source to the receiver reduces to two point-to-point throughput maximization problems. This means the source uses the algorithm of \cite[Section IV]{rezaee2015optimal} to transmit data to the relay, and the relay also uses the same algorithm to send data to the receiver. Now, we intuitively show that this separation is not trivial and must be proved. The main reason is that the optimal offline algorithm in \cite[Section IV]{rezaee2015optimal} maximizes the amount of transmitted data in $t=T$, and not for all $t\in (0,T]$. In the following, we present an example to further clarify this issue.
Assume that in the source we have $E_{s}(t)=5\times (t-1)^{3}+5$, and a lot of bits of data have been buffered at the beginning of transmission. Also, assume that $T=1$ and $r_{sr}(p)=\log(1+p)$. Based on optimal offline algorithm of \cite[Section IV]{rezaee2015optimal} the data transmitted from the source to the relay, at the end of transmission time ($t=T=1$), can be maximized if we use a fixed transmitted power $p=5$ or equally the transmission data rate $r=2.58$. This means that at the end of transmission, we have $2.58$ units of data at the relay. Now assume that we have a harvested energy curve in the relay such that cannot transmit $2.58$ units of data by $T=1$. Thus, we cannot transmit all the $2.58$ units of data to the receiver. The question is that what amount of data (and in which fashion) should arrive at the relay that can be matched to the harvested energy at the relay to maximize the transmitted data to the receiver. Maybe there is a policy which transmits for example $2.5$ units of data to the relay, but is more suitable for matching with harvested energy curve in the relay ($E_{r}(t)$).
%Note that in Theorem \ref{V.1.}, our goal does not consider the energy efficiency in the source. In Theorem \ref{V.1.}, we just want to maximize the amount of data to the receiver. In this remark we are explaining that finding the policy that maximize the transmitted data to the receiver is not trivial.
\end{rem}

Note that our goal is to find a policy which maximizes the throughput while minimizes the utilized energy in the source and relays. Theorem \ref{V.1.} shows that it is optimal to use the proposed algorithm in \cite{rezaee2015optimal} in the source (which obtains $B_{sr,s}^{*}(t)$), and then to consider this curve as an arrival data curve in the relay and start another round of proposed algorithm in \cite{rezaee2015optimal}. This procedure gives the optimal transmitted data curve at the relay.
We note that $B_{s}^{*}(t)$ is not equal to $B_{sr,s}^{*}(t)$ in general. This may be the case in the cases that the harvested energy in the relay is the bottleneck. In these cases, it is not necessary for the source to transmit as much as data it can, i.e., the source can transmit less data to the relay in a more efficient way to utilize energy without reducing performance in the relay. Formally, this means that there exists a $\tilde{B}_{sr}(t)$ such that $\int_{0}^{T}(\tilde{B}_{sr}(t)- B_{sr,s}^{*}(t))^{2}dt$, $B_{r}^{*}(t)=B_{rd,B_{sr,s}^{*}}^{*}(t)=B_{rd,\tilde{B}_{sr}}^{*}(t)$ and uses energy in a more efficient way at the source in these cases.
In the next theorem, we propose a policy, in which energy has been used in the most efficient manner while the \emph{optimal} transmitted data curve from the relay to the receiver proposed in Theorem \ref{V.1.} is feasible.

\begin{theorem}\label{sourceenergy}
(Source optimal policy) Assume that $l(t)$ is the tangent line to the curve $B_{sr,s}^{*}(t)$ passing through the point
$(T,B_{rd,B_{sr,s}^{*}}^{*}(T))$ as illustrated in Fig. \ref{sourceoptimalpolicy}. Also, Assume that $T_{1}=\max~\{t~|~l(t)=B_{sr,s}^{*}(t),~ 0\leq t\leq T\}$. The optimal policy for the source is as follows.
\begin{align}\label{sourcepolicy}
\hat{B}_{sr}(t)=\left\{\begin{matrix}
B_{sr,s}^{*}(t)& 0\leq t\leq T_{1} \\
l(t) & T_{1}\leq t\leq T
\end{matrix}\right.
\end{align}
\end{theorem}
 \begin{proof}[Outline of the proof]
We prove this theorem in three steps: 1) We show that $\hat{B}_{sr}(t)$ is a feasible policy in the source. 2) We show that if source transmits $\hat{B}_{sr}(t)$, then $B_{rd,B_{sr,s}^{*}}^{*}(t)$ is feasible in the relay. 3) We prove that for all feasible transmitted data curve which transmit the amount of $\hat{B}_{sr}(T)=B_{rd,B_{sr,s}^{*}}^{*}(T)$ data, $\hat{B}_{sr}(t)$ consumes minimum energy in the source. For the detailed proof, see Appendix \ref{apsourceenergy}.
\end{proof}
 \begin{figure}
                        \centering
                        \includegraphics[width=2.8in]{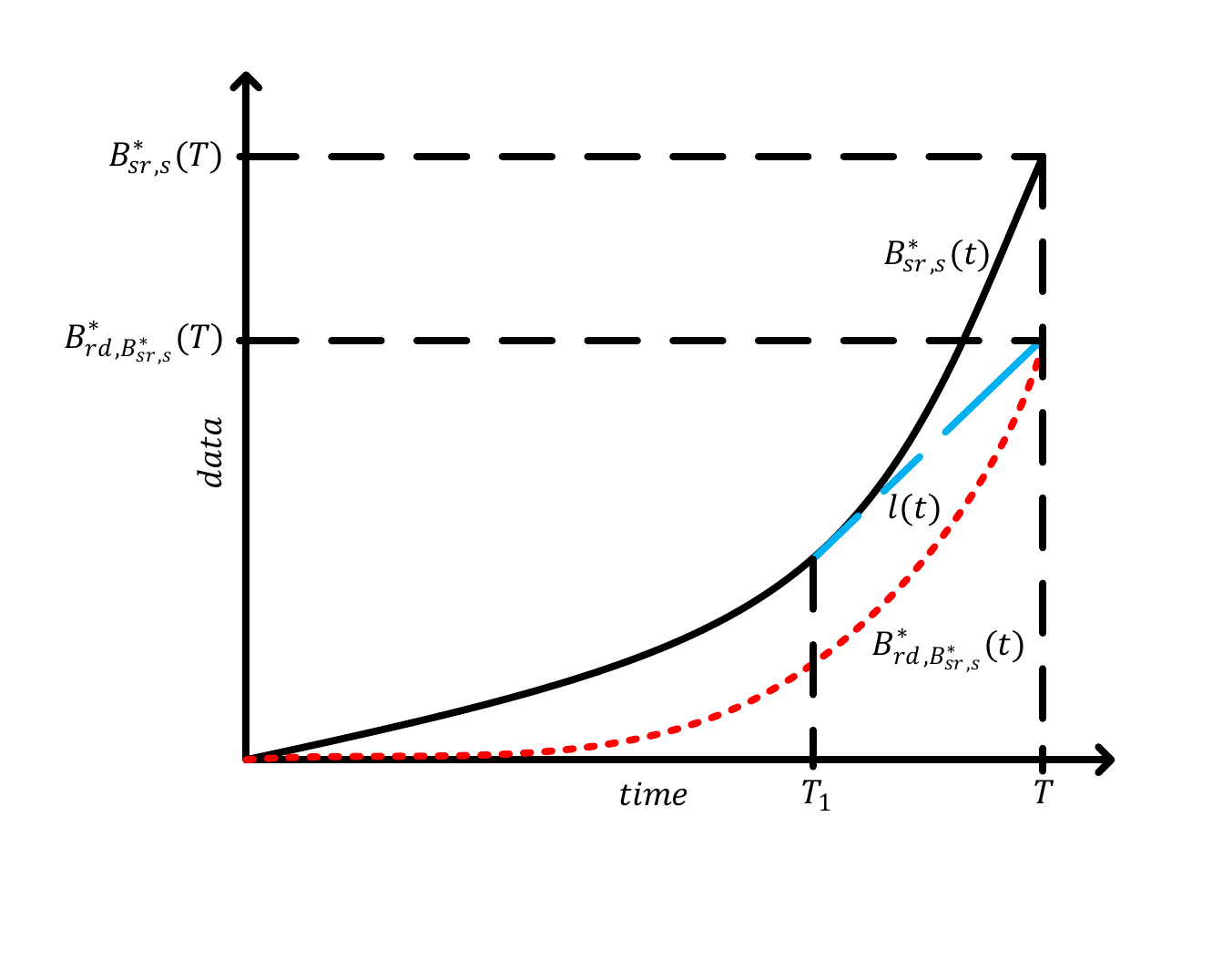}
                        \caption{The optimal policy in the source}
                        \label{sourceoptimalpolicy}
                        \end{figure}
\begin{coro}\label{1}
Theorems \ref{V.1.} and \ref{sourceenergy} can be extended to $n$ relays as follows: The source transmits maximum amount of data by applying the proposed algorithm in \cite[Section IV]{rezaee2015optimal} to the point-to-point throughput maximization problem (for the source-first relay link), the first relay sends maximum amount of data to the second relay by the same algorithm and this procedure repeats until the receiver. Then, using the proposed policy in Theorem \ref{sourceenergy}, we minimize the utilized energy in $(n-1)$-th relay. And after that using the new optimal transmitted data curve in the $(n-1)$-th relay, we minimize the utilized energy in the $(n-2)$-th relay. This process continues till the source.
 \end{coro}

%Now, we provide an example to show that that the
%\textbf{An example:}
% Assume that harvested energy curves in the source and the relay nodes are $E_{s}(t)=e^{t}-1$, $E_{r}(t)=2e^{t}-2$, respectively and $r_{sr}(p)=r_{rd}(p)=\frac{1}{2}\log(1+p)$ in which logarithm is in base 2. We want to maximize the throughput from the source to the destination. Using energy causality and convexity of $E_{s}(t)$ based on \cite[Section IV]{rezaee2015optimal} instantaneous arrival data at the relay is maximized in every $t\in[0,1]$ if $E_{sr}(t)=E_{s}(t)$. Thus, the optimal arrival data at the relay is $B_{sr,s}^{*}(t)=\int_{0}^{t}\frac{1}{2}\log(1+\frac{d}{dt^{'}}E_{s}(t^{'}))dt^{'}$ which is a continuous curve. Now, the problem reduces to a single-user throughput maximization problem in the relay node with harvested energy curve $E_{r}(t)$, and arrival data curve $B_{sr,s}^{*}(t)$.
\subsection{Completion Time Minimization}
In this subsection, we investigate the time minimization problem to transmit $B_{0}$ amount of data to the receiver in a multi-hop channel. We remark that the results of this section can be easily reduced to the point-to-point channel (which was not studied for the continuous model before). We can formulate the problem as follows:
 \begin{align}T_{\textrm{off}}= \min~ T~~~~~~~~~~~~~~~~~~~~~~\\
 s.t.~ \int^{T}_{0} r_{rd}(p_{rd}(t))dt=B_{0},~\textrm{ and }
 \eqref{43}-\eqref{46}.
 \end{align}
 \begin{lemma}\label{V.2}
 $D^{(MH)}(t)$ in \eqref{42} is nondecreasing. Also if $\lim_{p\to \infty}\frac{r_{rd}(p)}{p}=0$, then $D^{(MH)}(t)$ is continuous.
 \end{lemma}
 \begin{proof}
   The proof of the first part is trivial and is omitted for brevity. For the second part, let's define $E_{r,(t)}^{*}(t_{0})$ as the amount of optimal transmitted energy curve in the relay for the deadline $T=t$ in time $t_{0}$, and $0<A_{0}$ be a finite real number. Now, for $t\in(t_{0},t_{0}+\epsilon]$ with any $\epsilon \geq 0 ,t_{0}$ we have,
   \begin{align}
   D^{(MH)}(t_{0})&\leq D^{(MH)}(t)\stackrel{(a)}{\leq} D^{(MH)}(t_{0})+(t-t_{0}) r_{rd}(\frac{A(t)}{t-t_{0}})\nonumber\\
   &\stackrel{(b)}{\leq} D^{(MH)}(t_{0})+(t-t_{0}) r_{rd}(\frac{A_{1}(t)}{t-t_{0}})\nonumber\\
   &\stackrel{(c)}{<} D^{(MH)}(t_{0})+(t-t_{0}) r_{rd}(\frac{A_{2}(t)}{t-t_{0}})\label{9} ,
   \end{align}
   where (a) follows by setting $A(t)=E_{r}(t)-E_{r,(t)}^{*}(t_0)$ and the concavity of $r_{rd}(p)$, (b) follows from $A_{1}(t)=E_{r}(t)-E_{r,(t_{0}+\epsilon)}^{*}(t_0)$ and the fact that by increasing the deadline, $t$, the $E_{r,(t)}^*(t_0)$ decreases, and (c) follows from  $A_{2}(t)=A_{1}(t)+A_{0}$. Due to energy causality in relay, we have $A_{1}(t)\geq A(t)\geq 0$, $\forall t\in [t_{0},t_{0}+\epsilon]$, and thus $\lim_{t\to t_{0}}A_{2}(t)=E_{r}(t_{0})-E_{r,(t_{0}+\epsilon)}^{*}(t_{0})+A_{0}=A_{2}>0$. Substituting $p=\frac{A_{2}(t)}{t-t_{0}}$ in \eqref{9} and $\lim_{p \to \infty}\frac{r_{rd}(p)}{p}=0$ yields,
   \begin{align}
   &\lim_{t\to t_{0}^{+}} \left(D^{(MH)}(t_{0})+ (t-t_{0}) r_{rd}\big(\dfrac{A_{2}(t)}{t-t_{0}}\big)\right)\nonumber\\
   =&D^{(MH)}(t_{0})+ A_{2}\lim_{p\to \infty}\frac{r_{rd}(p)}{p}=D^{(MH)}(t_{0}).
   \end{align}
   From above, it is concluded $\lim_{t\to t_{0}^{+}}D^{(MH)}(t)=D^{(MH)}(t_{0})$. We can similarly prove that $\lim_{t\to t_{0}^{-}}D^{(MH)}(t)=D^{(MH)}(t_{0})$. Thus, $D^{(MH)}(t)$ is continuous.
   \end{proof}
In the next theorem, we show that the optimal solution to the completion time minimization problem is same as the optimal solution to its dual problem (i.e., throughput maximization problem) after fixing the deadline (minimum required time in this case).
\begin{theorem}\label{V.3}
If $\lim_{p\to \infty}\frac{r_{rd}(p)}{p}=0$, then the optimal algorithm for the completion time minimization problem is first to find the minimum completion time (i.e., $T_{\textrm{off}}$) and then to apply the proposed algorithm in theorems \ref{V.1.} and \ref{sourceenergy} for given deadline $T_{\textrm{off}}$.
 \end{theorem}
 \begin{proof}
Let $C=\left \{t:~~D^{(MH)}(t)=B_{0}  \right \}$. If $C\neq \emptyset$, let $T_{\min}=\min~C$.
Obviously, if $C\neq \emptyset$, exists a method to transmit amount of $B_{0}$ data until $T_{\min}$. Based on Lemma \ref{V.2}, if $T_{\textrm{off}}< T_{\min}$ holds, we have $D^{(MH)}(T_{\textrm{off}})<B_{0}$ which is a contradiction (because $T_{\textrm{off}}$ is the minimum completion time). Thus $T_{\textrm{off}}= T_{\min}$ and the optimal policy is the one proposed in theorems \ref{V.1.} and \ref{sourceenergy}.
%Conversely if there exists a time $T_{c}$ such that we can transmit $B_{0}$ amount of data until $T_{c}$, we get $B_{0}\leq D^{(MH)}(T_{c})$. Thus, Lemma \ref{V.2} concludes $C\neq \emptyset$.
 \end{proof}
 \begin{coro}\label{2}The result of Theorem \ref{V.3} is extended to $n$ relays with defining the throughput maximization problem as:
 \begin{small}
 \begin{eqnarray}
 D^{(MH)}(T)=\max_{p_{r_{0}r_{1}}(t),...,p_{r_{n}r_{n+1}}(t)} \int_{0}^{T}r_{r_{n}r_{n+1}}(p_{r_{n}r_{n+1}}(t))dt\nonumber~~~~~~~~~\\\nonumber
  s.t.~~\int_{0}^{t}p_{r_{i}r_{i+1}}(t^{'})\leq E_{i}(t),~0\leq t\leq T,~ 0\leq i\leq n~~~~~~~~~\\\nonumber
 \int_{0}^{t}r_{r_{0}r_{1}}(p_{r_{0}r_{1}}(t^{'}))dt^{'}\leq  B_{s}(t),~0\leq t\leq T~~~~~~~~~~\\\nonumber
     \int_{0}^{t}r_{r_{i+1}r_{i+2}}(p_{r_{i+1}r_{i+2}}(t^{'}))dt^{'}\leq ~~~~~~~~~~~~~~~~~~~~~\\\nonumber\int_{0}^{t}r_{r_{i}r_{i+1}}(p_{r_{i}r_{i+1}}(t^{'}))dt^{'},~0\leq t\leq T,~ 0\leq i\leq n,~~~~~~~~~~~
 \end{eqnarray}\end{small}where the subscripts $r_{0}$, $r_{i}$, and $r_{n+1}$ denote the source, $i$-th relay and the receiver, respectively. Also, $p_{r_{i}r_{i+1}}(t)$, $r_{r_{i}r_{i+1}}(p)$, and $E_{i}(t)$ are the transmitted power curve from $i$-th node to $(i+1)$-th node, the transmission rate of the channel as a continuous function of the transmitted power between $i$-th node to $(i+1)$-th node, and the harvested energy curve in $i$-th node, respectively.
 \end{coro}
\section{An Online Algorithm}\label{anlinealgorithm}
In this section we propose an online algorithm for the optimization problem \eqref{42}-\eqref{46}. In this case, we only have access to the causal information, i.e., the previous amounts of harvested energy and arrival data in the source and the relay. Our proposed online algorithm does not need to know the distributions of processes $E_{s}(t)$, $E_{r}(t)$ and $B_{s}(t)$. We show that the transmitted power curves in both source and relay are non-decreasing in the proposed online algorithm. This algorithm either transmits all arrived data in the source or consumes all the harvested energy in the source or relay. Further, we compare our online algorithm with a benchmark algorithm. The proposed online algorithm is as follows. We remark that all variables with a subscript $on$ refer to their corresponding variables in the online algorithm.

\textbf{Source}:
To determine the online transmitted power curve in the source, one must notice two factors, either the harvested energy or the arrival data. At each time, we specify the limiting factor in the source. Next, we set the transmitted power such that if at time $t$, the harvested energy is the limiting factor, then all the remaining energy in the source at time $t$ is consumed with fixed rate until time $T$. Otherwise, if the arrival data is the limiting factor, we set the transmitted power such that all the remaining data at time $t$ (in the source) is transmitted with a fixed data rate until time $T$.
\begin{align}\label{49}
p_{sr_{on}}(t)=\min
\left\lbrace r^{-1}_{sr}(\frac{B_{rem_{s}}(t)}{T-t+\epsilon}),\frac{E_{rem_{s}}(t)}{T-t+\epsilon}\right\rbrace 
\end{align}
where $B_{rem_{s}}(t)=B_{s}(t)-B_{sr_{on}}(t)$, $E_{rem_{s}}(t)=E_{s}(t)-E_{sr_{on}}(t)$, and $\epsilon>0$ is chosen to make the $p_{sr_{on}}(t)$ bounded.
The above procedure repeats when new energy is harvested or new data is arrived.

\textbf{Relay}:
The algorithm in the relay is different from the one in the source in the sense that there are two origins of data: the remaining data and the current arrival data from the source. Now the data limiting factor is their maximum. Thus, we first compare $r^{-1}_{rd}(\frac{B_{rem_{r}}(t)}{T-t+\epsilon})$ and $r^{-1}_{rd}(r_{sr}(p_{sr_{on}}(t)))$ and select the maximum one. This makes the algorithm be more efficient than only considering $r^{-1}_{rd}(\frac{B_{rem_{r}}(t)}{T-t+\epsilon})$ in \eqref{60}, because sometimes, such as the beginning of transmission, the amount of this term is low which limits the rate of the transmission. Hence, it costs energy.
We assume that $t_{1_{r}}$, $t_{2_{r}}$, ..., $t_{n_{r}}\in (0,T)$ are the all instants in which $p_{rd{on}}(t)$ switches between $r^{-1}_{rd}(\frac{B_{rem_{r}}(t)}{T-t+\epsilon})$, $r^{-1}_{rd}(r_{sr}(p_{sr_{on}}(t)))$ and $r^{-1}_{sr}(r_{rd}(\frac{E_{rem_{r}}(t)}{T-t+\epsilon}))$.

\begin{align}
&p_{rd_{on}}(t)=\nonumber\\
&\!\!\!\!\min\Big\{\max\big\{r^{-1}_{rd}(\frac{B_{rem_{r}}(t)}{T-t+\epsilon}), r^{-1}_{rd}(r_{sr}(p_{sr_{on}}(t)))\big\}, \frac{E_{rem_{r}}(t)}{T-t+\epsilon}\Big\},\label{60}
\end{align}
where $B_{rem_{r}}(t)=B_{sr_{on}}(t)-B_{rd_{on}}(t)$, $E_{rem_{r}}(t)=E_{r}(t)-E_{rd_{on}}(t)$, and (sufficiently small) $\epsilon>0$ is chosen to make the $p_{rd_{on}}(t)$ bounded.

The proposed online algorithms from source to the relay in this paper is chosen as that in \cite[Section V]{rezaee2015optimal}. The difference between direct generalization of algorithm in \cite[Section V]{rezaee2015optimal} and our proposed algorithm is in the relay's policy.

\begin{lemma}\label{convexcurve}
$p_{sr_{on}}(t)$ and $p_{rd_{on}}(t)$ are two non-decreasing functions.
\end{lemma}
\begin{proof}
This follows by simple extension of \cite[Lemma 27]{rezaee2015optimal}, and a detailed proof is thus omitted.
\end{proof}
\begin{lemma}\label{energydata}
In the proposed online algorithm, we either use all the energy at the relay or we transmit all data at the relay. In other words, if $\forall t\in(t_{n_{r}},T)$ we have $p_{rd_{on}}(t)=\frac{E_{rem_{r}}(t)}{T-t+\epsilon}$, then $\lim\limits_{\epsilon\to 0}E_{rd_{on}}(T)=E_{r}(T)$. Otherwise, $\lim\limits_{\epsilon\to 0}B_{rd_{on}}(T)=B_{sr_{on}}(T)$.
\end{lemma}
\begin{proof}
The proof can be easily obtained by extending the proof of \cite[Lemma 28]{rezaee2015optimal}.
\end{proof}
Now, we set a benchmark algorithm by applying the online algorithm of \cite[Section V]{rezaee2015optimal} twice (in both source and relay).
In the next theorem, we show that our proposed online algorithm transmits at least the amount of data which is transmitted using the benchmark algorithm. Moreover, we show that if the amounts of transmitted data to the receiver by using these two algorithms are the same, our proposed online algorithm uses less energy in the relay. The intuitive reason is that in the algorithm of \cite[Section V]{rezaee2015optimal} (for the point-to-point channel), the remaining energy and remaining data have been considered as limiting factors. For our proposed online algorithm, we have the same limiting factors in the source, however to improve efficiency in the relay, we add one more maximizing step to prevent waste of energy.
This provides a new opportunity for the relay in the cases where $r^{-1}_{rd}(\frac{B_{rem_{r}}(t)}{T-t+\epsilon})<r^{-1}_{rd}(r_{sr}(p_{sr_{on}}(t)))$.
\begin{theorem}
Consider $B_{rd_{on}}^{(2)}(t)$ as the transmitted data curve in relay derived by \eqref{60} and $B_{rd_{on}}^{(1)}(t)$ as the transmitted data curve in relay if we use online algorithm in \cite[Section V]{rezaee2015optimal} for both source and relay ($p_{rd_{on}}^{(1)}(t)=\min\left\{r^{-1}_{rd}(\frac{B^{(1)}_{rem_{r}}(t)}{T-t+\epsilon}), \frac{E^{(1)}_{rem_{r}}(t)}{T-t+\epsilon}\right\}$). We have $B_{rd_{on}}^{(1)}(t)\leq B_{rd_{on}}^{(2)}(t)$. In addition, if $B_{rd_{on}}^{(1)}(T)= B_{rd_{on}}^{(2)}(T)$, then $E_{rd_{on}}^{(2)}(T)\leq E_{rd_{on}}^{(1)}(T)$.
\end{theorem}
\begin{proof}
To prove the first part, we use proof by contradiction. Let exist some $t\in [0,T]$, where $B_{rd_{on}}^{(1)}(t)> B_{rd_{on}}^{(2)}(t)$. Assume that $t_{a}$ is the first point, for which there exists $t_{b}>t_{a}$ such that $\forall t\in(t_{a},t_{b})$, $B_{rd_{on}}^{(1)}(t)> B_{rd_{on}}^{(2)}(t)$. Therefore, there exists an interval $(t_{a},t_{a}+\epsilon)$ in which $p_{rd_{on}}^{(2)}(t)<p_{rd_{on}}^{(1)}(t)$, because the transmitted power curves in both algorithms are piecewise continuous and $r_{rd}(p)$ is increasing. Define $t_{c}$ as a point in $(t_{a},t_{a}+\epsilon)$. Similar to the Step 2 in proof of Theorem \ref{V.1.}, we obtain $E_{rd_{on}}^{(2)}(t_{c})<E_{rd_{on}}^{(1)}(t_{c})$ and this results in:
\begin{align}\label{energyrate}
\frac{E^{(1)}_{rem_{r}}(t_{c})}{T-t_{c}+\epsilon}<\frac{E^{(2)}_{rem_{r}}(t_{c})}{T-t_{c}+\epsilon},
\end{align}
Using the same policy in source, we get $B_{sr_{on}}^{(2)}(t)=B_{sr_{on}}^{(1)}(t)$, which is in accompany with $B_{rd_{on}}^{(2)}(t_{c})<B_{rd_{on}}^{(1)}(t_{c})$ results in  $B^{(1)}_{rem_{r}}(t_{c})<B^{(2)}_{rem_{r}}(t_{c})$. Thus
 \begin{align}\label{datarate}
  \frac{B^{1}_{rem_{r}}(t_{c})}{T-t_{c}+\epsilon}<\max\left\lbrace  r^{-1}_{rd}(\frac{B^{2}_{rem_{r}}(t_{c})}{T-t_{c}+\epsilon}), r^{-1}_{rd}(r_{sr}(p_{sr_{on}}(t_{c})))\right\rbrace.
  \end{align}
Combining \eqref{energyrate} and \eqref{datarate} result in $p_{rd_{on}}^{(1)}(t)<p_{rd_{on}}^{(2)}(t)$ which is a contradiction. Thus, $B_{rd_{on}}^{(1)}(t)\leq B_{rd_{on}}^{(2)}(t)$ holds $\forall t\in [0,T]$. The proof of second part of this theorem can be easily derived by using lemma \ref{zafar}.
\end{proof}
\section{Simulations}
In this section, we provide some numerical examples. We show that discretizing reduces the efficiency in a two-hop channel. Also, we compare the utilized energy between two offline algorithms: one of them only maximizes the throughput, but the other one maximizes the throughput while minimizing the utilized energy in both source and relay. In addition, we simulate the proposed online algorithm to confirm the obtained results.
In all simulations, we consider a band-limited additive white Gaussian noise channel with bandwidth $W=1$ Hz for both hops. Also, the channel gain divided by the noise power spectral density multiplied by the bandwidth is assumed to be $1$. Hence, we have, $r_{sr}(p)=r_{rd}(p)=\log(1+p)$, where the logarithm is in base $2$.

    \begin{figure*}
     \centering
     \includegraphics[width=7in]{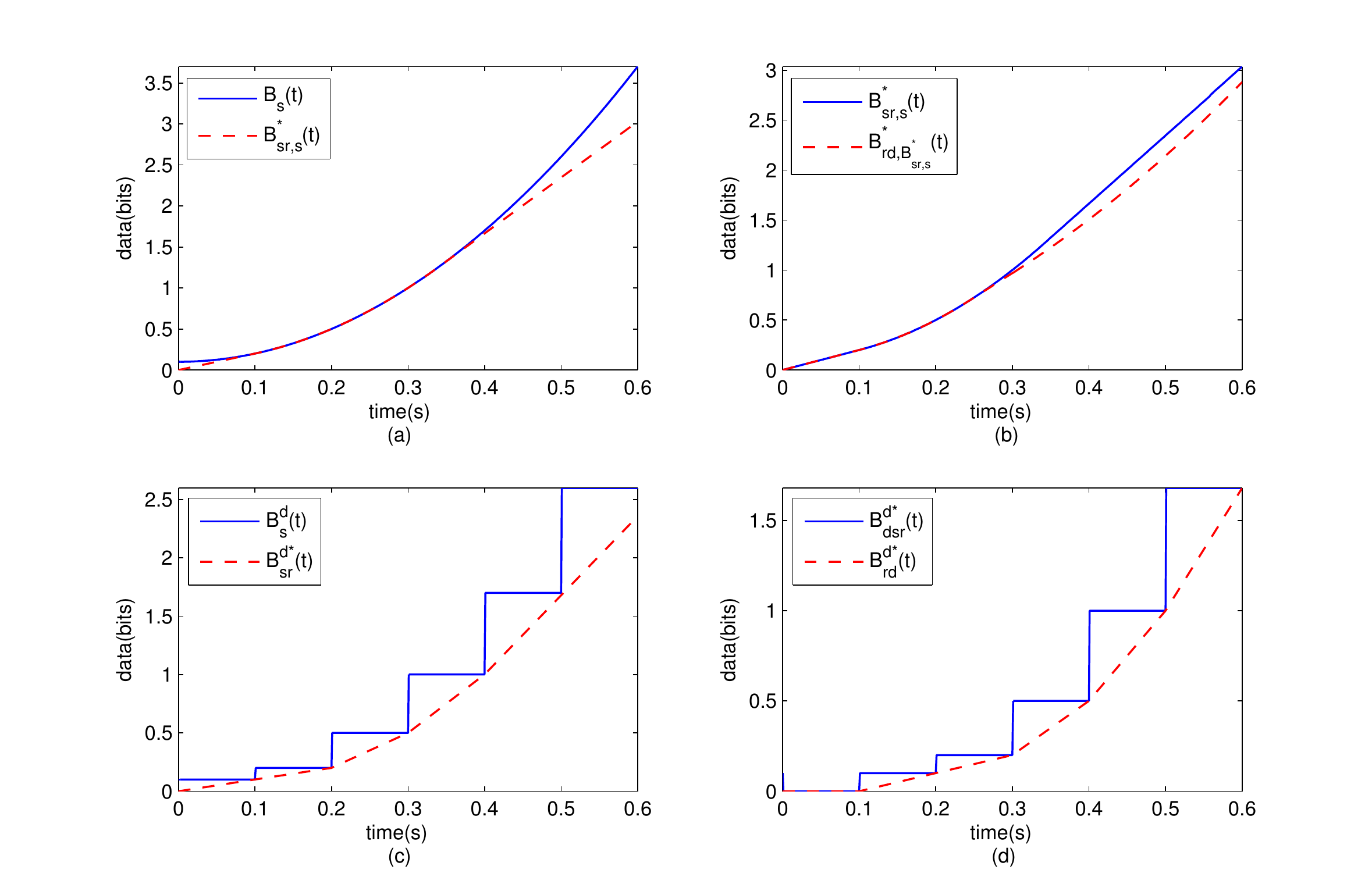}
     \caption{The effect of discretization on the optimal transmission policies. (a),(b): Transmitted data curves in source and relay (before discretizing), (c),(d) Transmitted data curves in source and relay (after discretizing)}
     \label{fig}
     \end{figure*}
     \begin{figure*}
       \centering
       \includegraphics[width=7in]{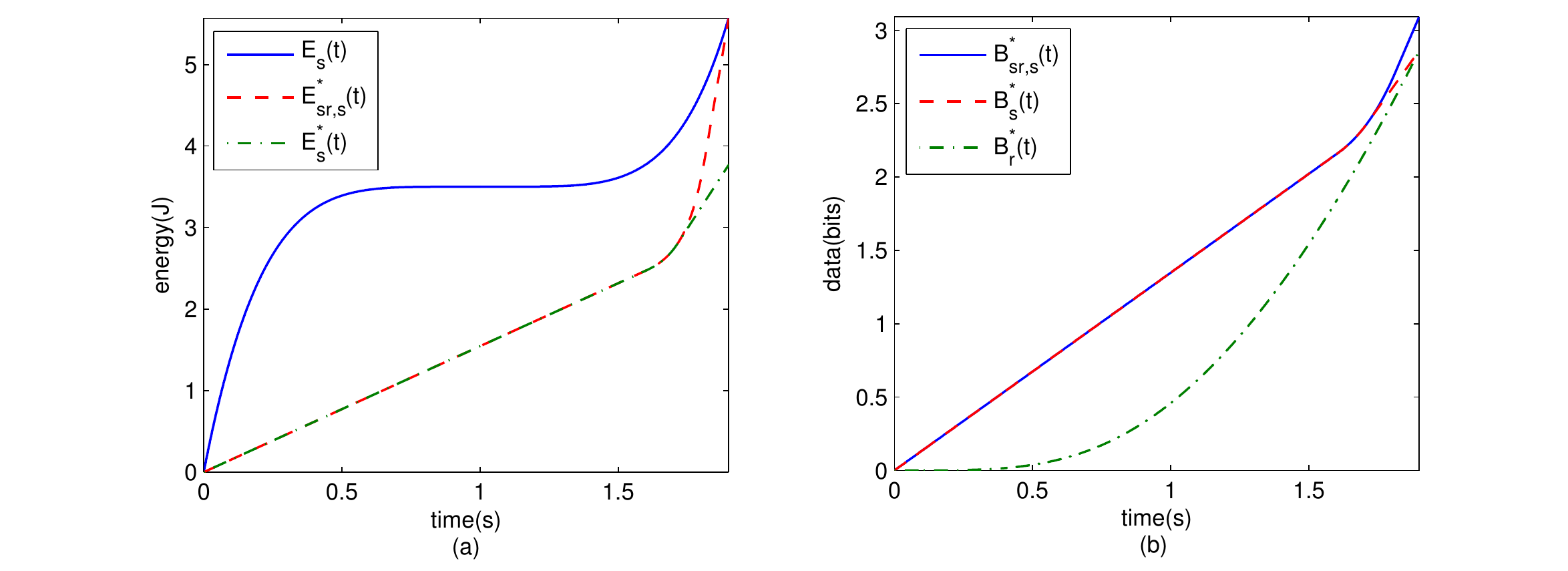}
       \caption{The optimal throughput maximization algorithms with and without energy minimization in source. (a) Energy curves. (b) Data curves.}
       \label{fig1}
       \end{figure*}
        \begin{figure*}
            \centering
            \includegraphics[width=7in]{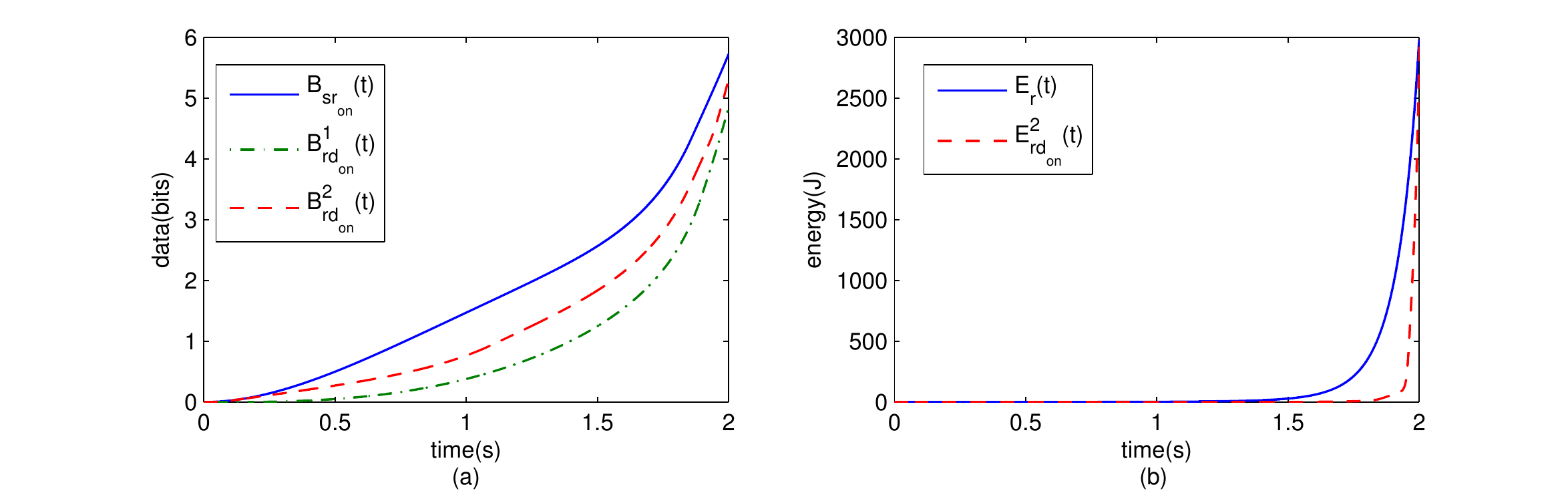}
            \caption{(a) illustrates the transmitted data curves in the source and relay using the proposed online algorithm and \cite[Section V]{rezaee2015optimal}. (b) illustrates the transmitted energy curve in relay using the proposed online algorithm.}
            \label{fig4}
            \end{figure*}
For the first example, we assume that $E_{s}(t)=100\times t^{2}+1$, $E_{r}(t)=0.5\times \exp(7t)-0.5$ J, $B_{s}(t)=10\times t^{2}+0.1$ bits and $T=0.6$ s. Fig. \ref{fig} (a) and Fig. \ref{fig} (b) illustrate the optimal transmitted data curves (at the source and relay, respectively) for the throughput maximization problem without minimizing energy in the source. Fig. \ref{fig} (c) and Fig. \ref{fig} (d) show the above curves after discretizing harvested energy and arrival data. In Fig. \ref{fig}, the $B_{s}^{d}(t)$ and $B_{dsr}^{d*}(t)$ are the discretized versions of the $B_{s}(t)$ and $B_{sr}^{d*}(t)$. Also, the results in Fig. \ref{fig} (c) and Fig. \ref{fig} (d) are obtained by discretizing $E_{s}(t)$ and $E_{r}(t)$.
As mentioned in Section I, it can be easily seen that the optimal offline algorithm with discretizing in harvested energy and arrival data transmits $1.68$ bits (compared to 2.88 bits in continuous model) which reduces the efficiency. This shows the necessity of investigating continuous model to achieve the optimal performance, instead of discretizing harvested energy and arrival data curves.

In the second example, we investigate the energy minimization constraint in the source. We assume that $E_{s}(t)=3.5\times (t-1)^{5}+3.5$, $E_{r}(t)=0.45\times t^{4}$ J, $B_{s}(t)=2 \times (t-1)^{5}+2$ bits and $T=1.9$ s. Note that $E_{s}^*(t)$ shows the optimal transmitted energy curve in the source with minimizing energy (our proposed algorithm) and $E_{sr,s}^*(t)$ shows the optimal transmitted energy curve in the source without minimizing energy (the point-to-point algorithm of \cite{rezaee2015optimal}). Their corresponding transmitted data curves are $B_{s}^{*}(t)$ and $B_{sr,s}^{*}(t)$, respectively, ( Fig. \ref{fig1} (b)).
Fig. \ref{fig1} shows that $E^{*}_{s,sr}(t)$ (without energy minimization) uses $5.5$ J to transmit $3.1$ bits of data to the relay while the relay can transmit at most $2.8$ bits of data to the receiver in the optimal policy. However, $3.8$ J energy in source is sufficient to transmit the same data to the receiver (based on $E_{s}^*(t)$ and $B_{s}^{*}(t)$).

Fig. \ref{fig4} shows the performance of our online algorithm and confirms the results in Section \ref{anlinealgorithm}.
We assume that $E_{s}(t)=80\times (t-1)^{3}+80$, $E_{r}(t)=\exp(t^{3})$ J, $B_{s}(t)=3.5\times (t-1)^{3}+3.5$ bits, $\epsilon=0.00001$ and $T=2$ s. As illustrated in Fig. \ref{fig4} (a), the transmitted data curve for the proposed online algorithm ($B^2_{rd_{on}}(t)$) transmits more data at any time than the direct extension of the online algorithm of \cite[Section V]{rezaee2015optimal} ($B^1_{rd_{on}}(t)$). Thus, the proposed online algorithm is more efficient than the benchmark one. Fig. \ref{fig4} confirms that the transmitted energy and data curves in the proposed online algorithm are convex (Lemma \ref{convexcurve}). Also, Fig. \ref{fig4} (b) shows that in the proposed online algorithm $E_{rd_{on}}^{2}(T)\simeq E_{r}(T)$ which confirms Lemma \ref{energydata}.

\section{Conclusion}
This paper considered an EH system with general arrival data and harvested energy curves in an FD multi-hop channel. An offline algorithm was proposed for the general model (including both discrete and continuous models), and was shown to be optimal. We remark that we considered energy minimization in the source and relays, preventing waste of energy in addition to the throughput maximization problem. In fact, the proposed optimal offline algorithm obtains a policy which uses the minimum needed energy in each node transmitting the maximum amount of data to the receiver. In other words, changing the policy at least one of nodes results in either more energy used or less total transmitted data to the receiver (compared with the proposed algorithm). 
In addition, we proposed an online algorithm and showed that it outperforms the direct extension of the online algorithm of \cite[Section V]{rezaee2015optimal} in relay. Our numerical examples confirmed the mathematically proven results. 
%For more practical models, it can be considered the finite energy buffer and/or data buffer (to provide quality-of-service) at the source and the relay nodes.
\appendix
\subsection{Proof of Theorem \ref{V.1.}}\label{throughputtwohop}
We start with some useful lemmas. The following lemmas state results for the point-to-point channel, where in the first lemma $p^{*}(t)$, $B(t)$ and $B^{*}(t)$ denote the optimal transmitted power curve, any feasible transmitted data curve and the optimal transmitted data curve, respectively, and in the second lemma $r(p)$ shows the transmission rate of the channel as a continuous function of the transmitted power.
\begin{lemma}\cite[Lemma 10]{rezaee2015optimal}\label{III.7.} If in an interval we have $\dfrac{d}{dt}p^{*}(t)\neq 0$, then $B(t)\leq B^{*}(t)$.
\end{lemma}
\begin{lemma}\cite[Lemma 23]{rezaee2015optimal}\label{zafar} Let $B_{1}(t)$ and $B_{2}(t)$ be two distinct feasible transmitted data curves and $B_{1}(t)> B_{2}(t)$ in the interval $(a,b)$ and $B_{1}(t)=B_{2}(t)$ at $t=a$ and $t=b$. If $B_{1}(t)$ is a convex function and $B_{1}(t)$ and $B_{2}(t)$ increase monotonically in $t$, then:
\begin{align} \int_{a}^{b}r^{-1}(\dfrac{d}{dt}B_{1}(t))dt< \int_{a}^{b}r^{-1}(\frac{d}{dt}B_{2}(t))dt. \end{align}
\end{lemma}

To clarify the proof we break the proof into four main steps. We first prove that for any feasible $B_{sr}(t)$, we have: $B_{rd,B_{sr}}^{*}(t)\leq B_{rd,B_{sr,s}^{*}}^{*}(t)$. To prove, we use contradiction as well as the technique mentioned in \cite[Remark 21]{rezaee2015optimal}.

\textbf{Step 1.} Let's assume that $B_{rd,B_{sr}}^{*}(t)\leq B_{rd,B_{sr,s}^{*}}^{*}(t)$ does not hold for $t\in[0,T]$. Because of $B_{rd,B_{sr}}^{*}(0)=B_{rd,B_{sr,s}^{*}}^{*}(0)=0$ and continuity of the curves, there exists an interval such that $B_{rd,B_{sr}}^{*}(t)> B_{rd,B_{sr,s}^{*}}^{*}(t)$ for the first time.  In details, assume that $a$ is the first point, for which there is $b>a$ that $\forall t\in(a,b)$, $B_{rd,B_{sr}}^{*}(t)> B_{rd,B_{sr,s}^{*}}^{*}(t)$. Therefore, there exists an interval $(a,a+\epsilon)$ in which $p_{rd,B_{sr,s}^{*}}^{*}(t)<p_{rd,B_{sr}}^{*}(t)$ because the transmitted power curve is piecewise continuous.
Now, we use \cite[(32)]{rezaee2015optimal} for an arbitrary $t_{0}\in(a,a+\epsilon)$ as follows:
\begin{small}\begin{align}p_{rd,B_{sr}}^{*}(t_{0})=\min\bigg\{\inf_{t_{0}<x\leq T}r^{-1}_{rd}\left( \frac{B_{sr}(x)-B_{rd,B_{sr}}^{*}(t_{0})}{x-t_{0}}\right) ,\nonumber\\\inf_{t_{0}<x\leq T}\frac{E_{r}(x)-E_{rd,B_{sr}}^{*}(t_{0})}{x-t_{0}}\bigg\}~~~~~~~\nonumber\\
 p_{rd,B_{sr,s}^{*}}^{*}(t_{0})=\min\bigg\{\inf_{t_{0}<x\leq T}r^{-1}_{rd}\left( \frac{B_{sr,s}^{*}(x)-B_{rd,B_{sr,s}^{*}}^{*}(t_{0})}{x-t_{0}}\right) ,\nonumber\\\inf_{t_{0}<x\leq T}\frac{E_{r}(x)-E_{rd,B_{sr,s}^{*}}^{*}(t_{0})}{x-t_{0}}\bigg\}.~~~~~~~
 \end{align}\end{small}

\textbf{Step 2.} In this step we show that,
 \begin{align}\label{energys}
   \inf_{t_{0}<x\leq T}\frac{E_{r}(x)-E_{rd,B_{sr}}^{*}(t_{0})}{x-t_{0}}<\inf_{t_{0}<x\leq T}\frac{E_{r}(x)-E_{rd,B_{sr,s}^{*}}^{*}(t_{0})}{x-t_{0}}.
  \end{align}

First, note that
 \begin{align}\label{ener}
 E_{rd,B_{sr,s}^{*}}^{*}(t_{0})=E_{rd,B_{sr,s}^{*}}^{*}(a)+\int_{a}^{t_{0}}p_{rd,B_{sr,s}^{*}}^{*}(t)dt\nonumber\\
 E_{rd,B_{sr,s}}^{*}(t_{0})=E_{rd,B_{sr,s}}^{*}(a)+\int_{a}^{t_{0}}p_{rd,B_{sr,s}}^{*}(t)dt.
 \end{align}
Based on Lemma \ref{zafar} and $B_{rd,B_{sr}}^{*}(t)\leq B_{rd,B_{sr,s}^{*}}^{*}(t),\; \forall t\in[0,a]$, we have:
\begin{align}\label{ener2}
E_{rd,B_{sr,s}^{*}}^{*}(a)\leq E_{rd,B_{sr}}^{*}(a)
\end{align}
Based on \eqref{ener}, \eqref{ener2} and $p_{rd,B_{sr,s}^{*}}^{*}(t)<p_{rd,B_{sr}}^{*}(t),\forall t\in(a,t_{0})$, we have $E_{rd,B_{sr,s}^{*}}^{*}(t_{0})<E_{rd,B_{sr}}^{*}(t_{0})$. This proves \eqref{energys}.

\textbf{Step 3.} In this step, we want to show:
\begin{align}\label{datas}
 \inf_{t_{0}<x\leq T}r^{-1}_{rd}\left( \frac{B_{sr}(x)-B_{rd,B_{sr}}^{*}(t_{0})}{x-t_{0}}\right) <\nonumber\\\inf_{t_{0}<x\leq T}r^{-1}_{rd}\left( \frac{B_{sr,s}^{*}(x)-B_{rd,B_{sr,s}^{*}}^{*}(t_{0})}{x-t_{0}}\right) .
\end{align}
If there exists a point $x_{c}>t_{0}$ such that $B_{sr,s}^{*}(x_{c})<B_{sr}(x_{c})$, then based on Lemma \ref{III.7.} there exists $\epsilon_{1}>0$ and $\epsilon_{2}>0$ such that
$B_{sr,s}^{*}(t)$ is linear in $(x_{c}-\epsilon_{1},x_{c}+\epsilon_{2})$, and
\begin{align}\label{cond1}
B_{sr,s}^{*}(t)<B_{sr}(t) \textrm{ for } t\in (x_{c}-\epsilon_{1},x_{c}+\epsilon_{2}).
\end{align}
In addition, since  $B_{sr,s}^{*}(T)\geq B_{sr}(T)$ and $B_{sr,s}^{*}(0)= B_{sr}(0)$, we can find some $\epsilon_{1}>0,\epsilon_{2}>0$ to satisfy:
\begin{align}
B_{sr,s}^{*}(x_{c}-\epsilon_{1})=B_{sr}(x_{c}-\epsilon_{1}),\nonumber\\
B_{sr,s}^{*}(x_{c}+\epsilon_{2})=B_{sr}(x_{c}+\epsilon_{2}).\label{equal}
\end{align}
The inequality in \eqref{cond1} results in: $\frac{B_{sr,s}^{*}(t)-B_{rd,B_{sr}}^{*}(t_{0})}{t-t_{0}}<\frac{B_{sr}(t)-B_{rd,B_{sr}}^{*}(t_{0})}{t-t_{0}}$.
Now, we have two cases:

(i) $\frac{B_{sr,s}^{*}(x_{c}-\epsilon_{1})-B_{rd,B_{sr}}^{*}(t_{0})}{x_{c}-\epsilon_{1}-t_{0}}\leq\frac{B_{sr,s}^{*}(x_{c}+\epsilon_{2})-B_{rd,B_{sr}}^{*}(t_{0})}{x_{c}+\epsilon_{2}-t_{0}}$, and

(ii) $\frac{B_{sr,s}^{*}(x_{c}+\epsilon_{2})-B_{rd,B_{sr}}^{*}(t_{0})}{x_{c}+\epsilon_{2}-t_{0}}<\frac{B_{sr,s}^{*}(x_{c}-\epsilon_{1})-B_{rd,B_{sr}}^{*}(t_{0})}{x_{c}-\epsilon_{1}-t_{0}}$
(illustrated in Fig.s \ref{firstcase} and \ref{secondcase}, respectively).

Because of linearity of $B_{sr,s}^{*}(t)$ in $(x_{c}-\epsilon_{1},x_{c}+\epsilon_{2})$, for the cases (i) and (ii), we have
\begin{align}\label{casei}
 \frac{B_{sr,s}^{*}(x_{c}-\epsilon_{1})-B_{rd,B_{sr}}^{*}(t_0)}{x_{c}-\epsilon_{1}-t_{0}}\leq\frac{B_{sr,s}^{*}(t)-B_{rd,B_{sr}}^{*}(t_{0})}{t-t_{0}},\\ \frac{B_{sr,s}^{*}(x_{c}+\epsilon_{2})-B_{rd,B_{sr}}^{*}(t_{0})}{x_{c}+\epsilon_{2}-t_{0}}\leq\frac{B_{sr,s}^{*}(t)-B_{rd,B_{sr}}^{*}(t_{0})}{t-t_{0}},\end{align} for $t\in[x_{c}-\epsilon_{1},x_{c}+\epsilon_{2}]$ respectively.
  Therefore, for the case (i) we have,
  \begin{align}\label{cond2}
  &\frac{B_{sr}(x_{c}-\epsilon_{1})-B_{rd,B_{sr}}^{*}(t_{0})}{x_{c}-\epsilon_{1}-t_{0}}\stackrel{(a)}{=}\frac{B_{sr,s}^{*}(x_{c}-\epsilon_{1})-B_{rd,B_{sr}}^{*}(t_{0})}{x_{c}-\epsilon_{1}-t_{0}}\nonumber\\
  \stackrel{(b)}{\leq} &\frac{B_{sr,s}^{*}(t)-B_{rd,B_{sr}}^{*}(t_{0})}{t-t_{0}}\stackrel{(c)}{<}\frac{B_{sr}(t)-B_{rd,B_{sr}}^{*}(t_{0})}{t-t_{0}}~~~~~~
\end{align}
where (a) is due to \eqref{equal}, (b) comes from \eqref{casei}, and (c) follows from \eqref{cond1}. And, similarly, for the case (ii) follows:
\begin{align}\label{cond3}
   &\frac{B_{sr}(x_{c}+\epsilon_{2})-B_{rd,B_{sr}}^{*}(t_{0})}{x_{c}+\epsilon_{2}-t_{0}}=\frac{B_{sr,s}^{*}(x_{c}+\epsilon_{2})-B_{rd,B_{sr}}^{*}(t_{0})}{x_{c}+\epsilon_{2}-t_{0}}\nonumber\\
   \leq &\frac{B_{sr,s}^{*}(t)-B_{rd,B_{sr}}^{*}(t_{0})}{t-t_{0}}<\frac{B_{sr}(t)-B_{rd,B_{sr}}^{*}(t_{0})}{t-t_{0}}~~~~~~
 \end{align}
for $t\in(x_{c}-\epsilon_{1},x_{c}+\epsilon_{2})$. Now, let $m=\argmin\limits_{t_{0}<x\leq T}(\frac{B_{sr}(x)-B_{rd,B_{sr}}^{*}(t_{0})}{x-t_{0}})$. Due to \eqref{cond2} and \eqref{cond3}, $m \not\in(x_{c}-\epsilon_{1},x_{c}+\epsilon_{2})$ (and similarly $m$ does not belong to all intervals that $B_{sr,s}^{*}(t)<B_{sr}(t)$). Thus, $B_{sr}(m)\leq B_{sr,s}^{*}(m)$, and this results in \eqref{datas}.

 From \eqref{energys} and \eqref{datas}, we have $p_{rd,B_{sr}}^{*}(t_{0})<p_{rd,B_{sr,s}^{*}}^{*}(t_{0})$ which is a contradiction. Therefore, $B_{rd,B_{sr}}^{*}(t)\leq B_{rd,B_{sr,s}^{*}}^{*}(t)$ holds for $t\in[0,T]$.

\textbf{Step 4.} In this step, we show that if there exists $\bar{B}_{sr}(t)$ such that $B_{rd,B_{sr,s}^{*}}^{*}(T)=B_{rd,\bar{B}_{sr}}^{*}(T)$, then $E_{rd,B_{sr,s}^{*}}^{*}(T) < E_{rd,\bar{B}_{sr}}^{*}(T)$.
Assume that there is a feasible transmitted data curve $B_{rd,\bar{B}_{sr}}^{*}(t)$ such that $\int_{0}^{T}(B_{rd,B_{sr,s}^{*}}^{*}(t)-B_{rd,\bar{B}_{sr}}^{*}(t))^{2}dt\neq 0$ and $B_{rd,B_{sr,s}^{*}}^{*}(T)=B_{rd,\bar{B}_{sr}}^{*}(T)$, then Step 3 shows that $B_{rd,\bar{B}_{sr}}^{*}(t)\leq B_{rd,B_{sr,s}^{*}}^{*}(t)$ for $t\in [0,T]$ and based on Lemma \ref{zafar} we have $E_{rd,B_{sr,s}^{*}}^{*}(T) < E_{rd,\bar{B}_{sr}}^{*}(T)$. This completes the proof of Theorem \ref{V.1.}.
 \begin{figure}
            \centering
            \includegraphics[width=3in]{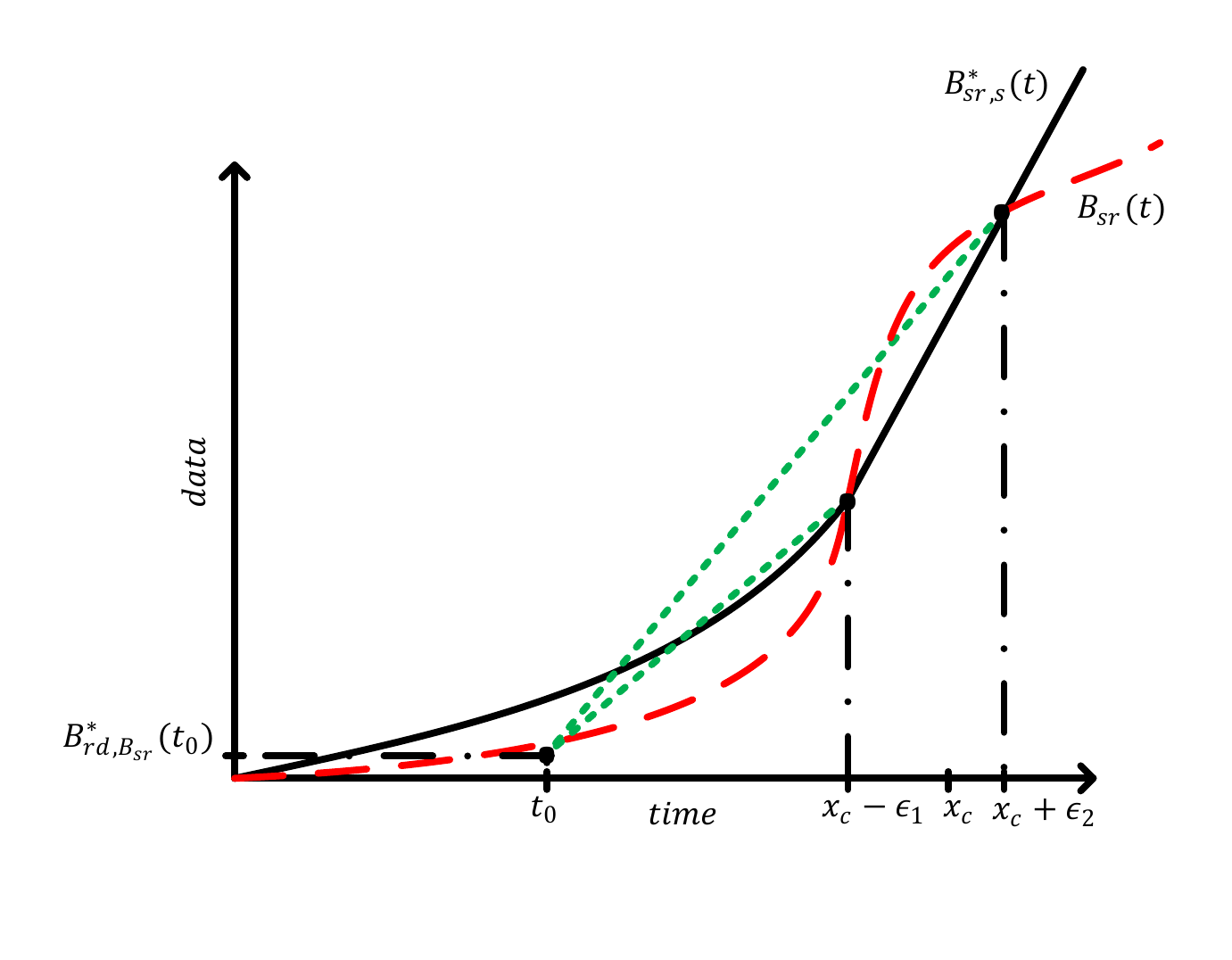}
            \caption{Case (i) in Step~3 (proof of Theorem \ref{V.1.}).}
            \label{firstcase}
            \end{figure}
             \begin{figure}
                        \centering
                        \includegraphics[width=3in]{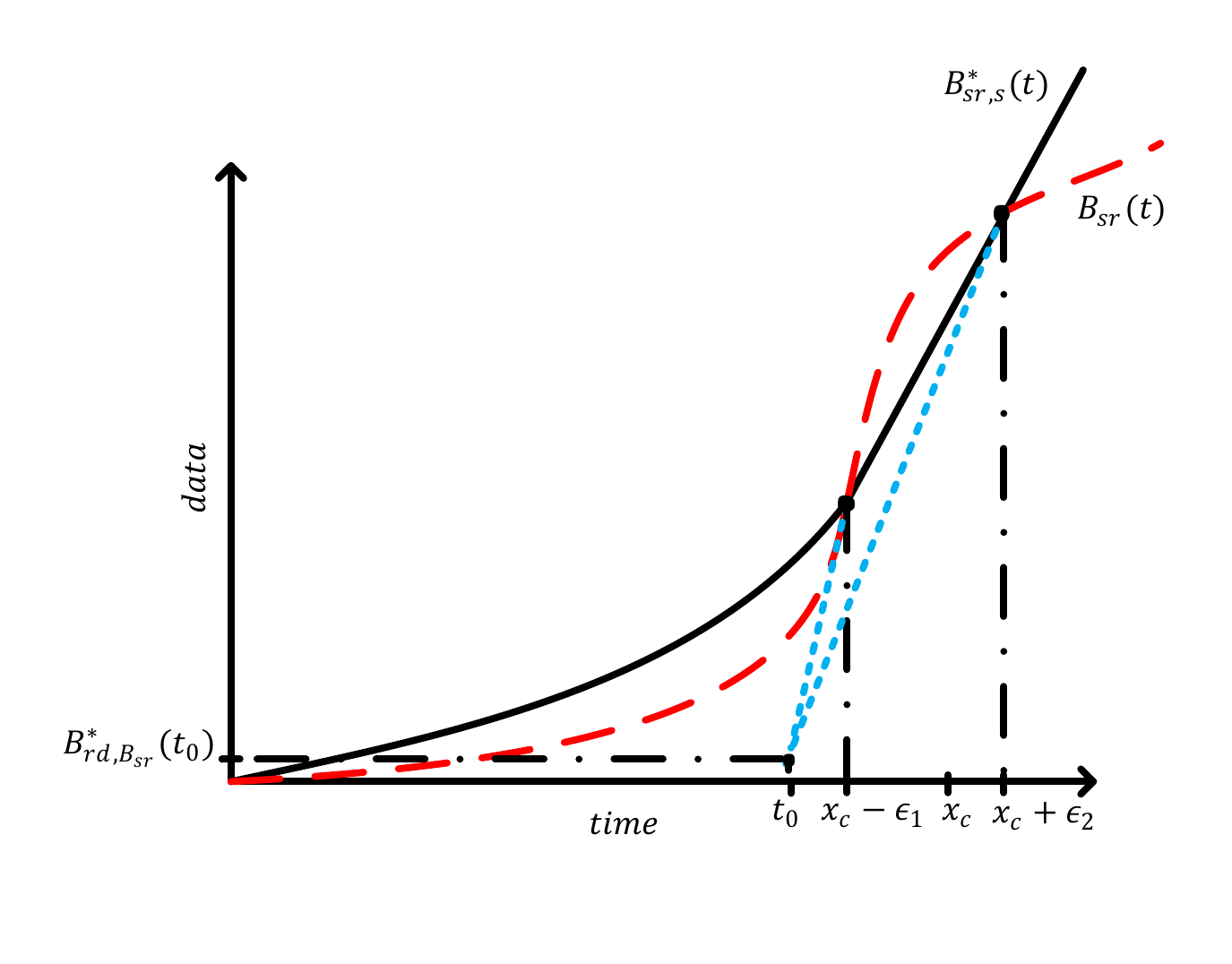}
                        \caption{Case (ii) in Step~3 (proof of Theorem \ref{V.1.}).}
                        \label{secondcase}
                        \end{figure}
\subsection{Proof of Theorem \ref{sourceenergy}}\label{apsourceenergy}
We prove this theorem in three steps: 1) We show that $\hat{B}_{sr}(t)$ is a feasible policy in the source. 2) We show that if source transmits $\hat{B}_{sr}(t)$, then $B_{rd,B_{sr,s}^{*}}^{*}(t)$ is feasible in the relay. 3) We prove that for all feasible transmitted data curve which transmit the amount of $\hat{B}_{sr}(T)=B_{rd,B_{sr,s}^{*}}^{*}(T)$ data, $\hat{B}_{sr}(t)$ consumes minimum energy in the source.

\textbf{Step 1.} In $[0,T_{1}]$ based on \eqref{sourcepolicy}, we have $\hat{B}_{sr}(t)=B_{sr,s}^{*}(t)$ and thus $\hat{B}_{sr}(t)$ is feasible. For $t\in [T_{1},T]$, $\hat{B}_{sr}(t)$ is linear and based on \cite[Lemma 4]{rezaee2015optimal}, $B_{sr,s}^{*}(t)$ is a convex curve. Thus, $\frac{d}{dt}l(t)\leq\frac{d}{dt}B_{sr,s}^{*}(t)$ holds for $t\in [T_{1},T]$, i.e., the slope of $l(t)$ is not greater than $\frac{d}{dt}B_{sr,s}^{*}(t)\lvert_{T^{+}_{1}}$. Therefore, energy and data causality hold for the proposed transmitted data curve in \eqref{sourcepolicy}.

\textbf{Step 2.} To prove this step, it is enough to show that $B_{rd,B_{sr,s}^{*}}^{*}(t) \leq \hat{B}_{sr}(t)$ holds for $t\in [0,T]$. In $[0,T_{1}]$, $\hat{B}_{sr}(t)=B_{sr,s}^{*}(t)$  results in $B_{rd,B_{sr,s}^{*}}^{*}(t)\leq \hat{B}_{sr}(t)$ for $t\in [0,T_{1}]$ (noting that $B_{rd,B_{sr,s}^{*}}^{*}(t)\leq B_{sr,s}^{*}(t)$). In $[T_{1},T]$, since $B_{rd,B_{sr,s}^{*}}^{*}(t)\leq B_{sr,s}^{*}(t)$ we have two cases:
i) $B_{rd,B_{sr,s}^{*}}^{*}(T_{1})=B_{sr,s}^{*}(T_{1})$: in this case, $l(t)$ connects $(T_{1}, B_{rd,B_{sr,s}^{*}}^{*}(T_{1}))$ and $(T, B_{rd,B_{sr,s}^{*}}^{*}(T))$. Due to the convexity of $B_{rd,B_{sr,s}^{*}}^{*}(t)$, we have $B_{rd,B_{sr,s}^{*}}^{*}(t) \leq \hat{B}_{sr}(t)$ in $[T_{1},T]$. ii) $B_{rd,B_{sr,s}^{*}}^{*}(T_{1})<B_{sr,s}^{*}(T_{1})$: $B_{rd,B_{sr,s}^{*}}^{*}(t) \leq \hat{B}_{sr}(t)$ follows from the first case $\forall t\in [T_1,T]$.

\textbf{Step 3.} Assume that $B_{sr}(t)$  is a feasible transmitted data curve such that $\int_{0}^{T}(B_{sr}(t)- \hat{B}_{sr}(t))^{2}dt\neq 0$ and transmits $\hat{B}_{sr}(T)$ amount of data at the end of transmission process. If $T=T_{1}$, then it is proved in \cite[Theorem~25]{rezaee2015optimal} that any $B_{sr}(t)$ uses more energy than $\hat{B}_{sr}(t)$ to transmit $B_{sr}(T)$ amount of data.

If $T_{1}< T$, then we have $\frac{d}{dt}p_{sr,s}^{*}(t)\lvert_{T_{1}}\neq 0$ (Otherwise, $T_{1}\neq\max~\{t~|~l(t)=B_{sr,s}^{*}(t),~ 0\leq t\leq T\}$).
Therefore, based on Lemma \ref{III.7.} we have $B_{sr}(T_{1})\leq \hat{B}_{sr}(T_{1})$. Thus if $\hat{B}_{sr}(t)<B_{sr}(t)$ holds in $(a,b)$, then based on Lemma \ref{III.7.} and noting that $\hat{B}_{sr}(t)$ is linear in $[T_{1},T]$ and $B_{sr}(T_{1})\leq \hat{B}_{sr}(T_{1})$, $\hat{B}_{sr}(t)$ is linear in $(a,b)$. Now, one can use Jensen's inequality to show that in the intervals in which $\hat{B}_{sr}(t)<B_{sr}(t)$, $B_{sr}(t)$ uses more energy.
Note that because of $l(t)$ is tangent line to the curve $B_{sr,s}^{*}(t)$, passing through the point $(T_{1},B_{sr,s}^{*}(T_{1}))$, and $B_{sr,s}^{*}(t)$ is convex, the slope of line $l(t)$ is greater than or equal to $\frac{d}{dt}B_{sr,s}^{*}(t)\lvert_{T^{-}_{1}}$. Since $B_{sr,s}^{*}(t)$ is convex, $\hat{B}_{sr}(t)$ is also convex.
Thus, in the intervals, in which $B_{sr}(t)<\hat{B}_{sr}(t)$ holds, $\hat{B}_{sr}(t)$ uses less energy than $B_{sr}(t)$ based on Lemma~\ref{zafar}. As a conclusion, $\hat{B}_{sr}(t)$ uses lowest energy among transmitted data curves which transmit $\hat{B}_{sr}(T)$ amount of data in the source.
\bibliographystyle{IEEEtranTCOM}
% argument is your BibTeX string definitions and bibliography database(s)
\bibliography{IEEEabrv}
\end{document}